\documentclass[12pt]{article}%
\usepackage{soul}
\usepackage{subfloat}
\usepackage{subcaption}
\usepackage{amssymb,fullpage}
\usepackage{amsfonts}
\usepackage{amsmath}
\usepackage{graphicx}
\usepackage[titletoc]{appendix} 
\usepackage[group-separator={,}]{siunitx}
\usepackage{hyperref}
\usepackage{cleveref}

\usepackage{enumerate}

\usepackage[usenames,dvipsnames,svgnames,table]{xcolor}
\usepackage[ruled]{algorithm2e} 

\SetAlFnt{\small}
\SetAlCapFnt{\small}
\SetAlCapNameFnt{\small}
\SetAlCapHSkip{0pt}
\IncMargin{-\parindent}
\usepackage{tikz}
\usetikzlibrary{arrows.meta}

 \numberwithin{equation}{section}

\usepackage{amsfonts, amsthm, thmtools,graphicx,latexsym,url,epsf, algpseudocode, verbatim, setspace, float, paralist, mathtools}

\DeclareCaptionLabelFormat{subfig}{\figurename #1~\arabic{figure}.\Alph{subfigure}:}
\usepackage{natbib}
\def\cite{\citep}
\bibpunct{[}{]}{,}{a}{}{;}

\usepackage{hyperref}
\hypersetup{
  colorlinks,
  citecolor=blue,
  linkcolor=blue,
  urlcolor=blue}

\newcommand{\bmath}{\begin{equation}}
\newcommand{\emath}{\end{equation}}
\newcommand{\bmathnn}{\begin{eqnarray*}}
\newcommand{\emathnn}{\end{eqnarray*}}

\declaretheorem[numberwithin=section]{theorem}
\declaretheorem[sibling=theorem]{lemma}
\declaretheorem[sibling=theorem]{proposition}

\declaretheorem[sibling=theorem]{corollary}

\def\setminus{-}
\def\1{{\bf{1}}}

\usepackage{color}

\usepackage{nomencl}
\makenomenclature

\begin{document}

\title{The Value of Excess Supply in Spatial Matching Markets\thanks{We thank Nick Arnosti, Itai Ashlagi, Yash Kanoria, David Kreps, Paul Milgrom, and several seminar participants for helpful comments.}}

\author{
Mohammad Akbarpour,\thanks{Graduate School of Business, Stanford University.  \protect\url{mohamwad@stanford.edu}
}
 \, Yeganeh Alimohammadi,\thanks{Department of Management Science and Engineering, Stanford University.  \protect\url{yeganeh@stanford.edu}
 } \\    Shengwu Li,\thanks{Department of Economics, Harvard University. \protect\url {shengwu_li@fas.harvard.edu}} \, Amin Saberi\thanks{Department of Management Science and Engineering, Stanford University.  \protect\url{saberi@stanford.edu}
 }
 }

\maketitle

\begin{abstract}
We study dynamic matching in a spatial setting. Drivers are distributed at random on some interval. Riders arrive in some (possibly adversarial) order at randomly drawn points. The platform observes the location of the drivers, and can match newly arrived riders immediately, or can wait for more riders to arrive. Unmatched riders incur a waiting cost $c$ per period. The platform can match riders and drivers, irrevocably. The cost of matching a driver to a rider is equal to the distance between them. We quantify the value of slightly increasing supply. We prove that when there are $(1+\epsilon)$ drivers per rider (for any $\epsilon > 0$), the cost of matching returned by a simple greedy algorithm which pairs each arriving rider to the closest available driver is $O(\log^3(n))$, where $n$ is the number of riders. On the other hand, with equal number of drivers and riders, even the \emph{ex post} optimal matching does not have a cost less than $\Theta(\sqrt{n})$. Our results shed light on the important role of (small) excess supply in spatial matching markets.
\end{abstract}



\thispagestyle{empty}


\onehalfspacing


\newpage

\setcounter{page}{1}

\section{Introduction}

In many real-world settings, we want to provide services when they are requested, by assigning nearby service providers.  Examples of such `spatial' matching problems include assigning drivers to riders in ride-hailing, assigning ambulances to medical emergencies, or (more figuratively) assigning workers to tasks that require varying skillsets.
These situations can be described as \textit{dynamic spatial matching problems}.  They have two key features.  First, requests arrive over time, and we are required to fulfil each request when it is made.  Second, each request (`rider') and each provider (`driver') has a geographic location, and nearby drivers are better than distant drivers.

One natural approach to spatial matching problems is to first balance supply and demand, and then to optimize the matching between drivers and riders.  The cost of hiring excess drivers is immediate and concrete, whereas the potential benefits are subtle and may depend on the matching algorithm.  Moreover, in some jurisdictions, ride-hailing platforms face a penalty for having excess drivers, because they are required to pay higher wages when the utilization rate is low.\footnote{\textit{New York City Considers New Pay Rules for Uber Drivers}, The New York Times, July 2 2018 and \textit{Seattle Passes Minimum Pay Rate for Uber and Lyft Drivers}, The New York Times, September 29 2020.}

In our model, $n$ `drivers' are uniformly distributed on the interval $[0,\ell]$, and  $n$ `riders' arrive over time, also uniformly distributed.  At any time, the platform match some driver to some rider, which removes both from the system.  Each unmatched rider accrues a waiting cost $c$ per unit time.  We start by assuming that the platform is obliged to serve all riders, and wants to minimize the total cost, \emph{i.e.} the total distance between riders and their matched drivers plus any waiting costs.

In this model, the optimal policy is potentially complex.  It depends on the locations of all the remaining drivers and riders, the arrival process of riders, and also on what the platform knows about the future.  Hence, as a tractable benchmark, suppose that the platform can perfectly predict the positions and the arrival times of the riders, and matches drivers to riders in such a way as to minimize the \textit{ex post} total cost.  This ``omniscient'' matching algorithm performs at least as well as any feasible algorithm. In particular, all drivers are present at the beginning, and the omniscient algorithm knows the optimal ex post matching.  Thus, the omniscient algorithm matches every rider upon arrival and pays no waiting costs.

To quantify the potential gains from excess supply, we compare the performance of the omniscient algorithm in balanced and unbalanced markets.  First, we prove that in a balanced market, with $n$ riders and $n$ drivers, the expected cost of this omniscient algorithm is $\Theta(\ell\sqrt{n})$.  Next, we consider an unbalanced market, with $n$ riders and $n(1 + \epsilon )$ drivers, for any positive constant $\epsilon$.  
We show that for any $\epsilon > 0$, adding $\epsilon n$ drivers causes the expected cost of the omniscient algorithm to fall from $\Theta(\ell\sqrt{n})$ to $O(\ell)$.  Hence, unbalancedness can potentially generate stark reductions in cost; the \emph{total} cost of the matching in an unbalanced market does not even rise with the market size $n$.

The omniscient benchmark is tractable but unrealistic.  What if the platform cannot see the future or has limited computational power? How much of the gain from increasing supply can we realistically achieve?  Consider the following greedy algorithm: Match each arriving rider to the nearest available driver.  The greedy algorithm does not exploit (and does not require) information about future arrivals, and it does not take into account the positions of any drivers except the nearest. Moreover, suppose riders arrive in an \textit{adversarial} order.  We prove that the greedy algorithm with $n(1+\epsilon)$ drivers has an expected matching cost of $O(\ell\log^3 (n))$.  It substantially outperforms the expected cost of the omniscient algorithm with $n$ drivers, and closes most of the gap to the omniscient algorithm with $n(1+\epsilon)$ drivers.  

Next we investigate whether the difficulty of the balanced-market problem is caused by the requirement that the platform serve all riders, even those who are far away.  We study a relaxed problem, in which the platform faces a balanced market but can leave a rider unmatched for a penalty $\nu$, which may depend on the market size $n$.  Of course, the problem is trivial without a lower bound for $\nu$.  For instance, if $\nu < \ell n^{-\frac{1}{2}}$, then for large $n$, the platform prefers matching no riders to matching every rider.  Therefore, suppose instead that $\nu \geq \ell n^{-\frac{1}{2} + \delta}$ for some positive constant $\delta$.  We prove that the expected cost of the omniscient algorithm is $\Omega(\ell n^\delta)$.  Hence, giving the platform the option to forego matching distant riders still leads to costs that are at least polynomial in $n$.

After presenting our theoretical results, we use simulations to investigate whether similar findings hold for small numbers of drivers and riders. To do so, we pose the following question: For a fixed number of riders $n$, how many excess drivers are needed in order for the greedy algorithm to beat the balanced-market omniscient algorithm? We find that for $n=25$, having 1 excess driver and greedily matching riders has a lower total distance cost than the omniscient algorithm.  For $n = 100$, $4$ excess drivers are enough, and for $n=1000$, and $13$ excess drivers are enough. 

Taken together, our results shed light on the importance of excess supply in spatial matching markets. Various real-world matching platforms have adopted policies with the goal of balancing supply and demand.  For instance, Uber asserts in an explanatory video about surge pricing that ``the main focus is trying to bring balance to the marketplace."\footnote{\url{https://www.uber.com/us/en/marketplace/pricing/surge-pricing/}, accessed March 25 2021.}  Similarly, Lyft states, ``Dynamic pricing is the main technology that allows us to maintain market balance in real-time."\footnote{\url{https://eng.lyft.com/dynamic-pricing-to-sustain-marketplace-balance-1d23a8d1be90}, accessed March 25 2021.} Our results suggest that these maxims should not be taken too literally.  Instead of exactly balancing supply and demand, mild excess supply can substantially improve performance, and even simple matching algorithms can realize these gains.

\subsection{Related work}

Online matching has a long history in theoretical computer science with many gems. 
\citealt{KarpVV} introduced the online bipartite matching problem. More than a decade later, and motivated by applications in online advertising, 
\citealt{MehtaSaberi} formulated the AdWords problem which has been studied extensively in its own right by \citealt{DevanurJainKleinberg,ManshadiGharanSaberi,goelmehta,MahdianSaberiNazer, HuangZhang2020}.
There is also a surge of interest in online matching in economics. In particular, 
\citealt{AkbarpourLi} studied a dynamic matching problem on stochastic networks, where agents arrive and depart over time and quantified the value of liquidity in such markets. 
\citealt{Baccara2020} studied optimal dynamic matching and thickness in a two-sided model.
Moreover, online matching models have been applied to multiple domains, including kidney exchange \citep{unver2010dynamic, AshlaghiKidney2013,AndersonAshlagi, ashlagi2019matching, akbarpour2020unpaired}, housing markets \citep{Leshno2019DynamicMI,blochHouy, arnosti2019not}, and ride-sharing \citep{OzkanWard, liu2019efficiency, castillo2020benefits}.

Online matching is also studied on metric spaces when the placement of the nodes and their arrivals are adversarial and the matching algorithm is deterministic in \citealt{KHULLER1994255,kalyanasundaram1993online, Antoniadis2015}, or randomized in \citealt{Meyerson2006RandomizedOA,Bansal07,NayyarRaghvendra17}, as well as the case where the arrival and the position of the nodes are chosen from a known distribution by \citealt{Wajc2019}. These results do not compare directly with the analysis of the current paper as they aim to find algorithms with the smallest competitive ratio, while we directly bound the cost of optimum matching as well as the cost of the solution of a particular (greedy) algorithm.  In fact, from the point of view of competitive analysis, it was shown  by \citealt{kalyanasundaram1991line,kalyanasundaram1993online} that  greedy can perform exponentially poorly in the worst case, and by \citealt{GAIRING201988} it  cannot perform better than $\Omega(n^{0.292})$
with random arrivals. \citealt{OMBM2016} studied greedy  in an experimental point of view and observed that greedy performs better than most of the known algorithms. The result of our study will give theoretical support to this observation.

The most relevant line of studies from a technical perspective is the (empirical) optimal transport problem.  There, the expected cost of the \emph{ex post} optimal matching is studied in Euclidean spaces with uniform random points by \citealt{Monvel2002AlmostSC,holroyd2020minimal,Ajtai1984OnOM,Trevisan2016},  and Poisson point processes by  \citealt{holroyd2020minimal,hoffman2006}. Also,  \citealt{FriezeReed90} study {\em non-bipartite} matching of $2n$ uniform random points on $[0,1]$ and show that the cost of greedy is $\Theta(\log(n))$.
The above results consider settings that correspond to balanced markets. To the best of our knowledge, our work is the first to give a constant bound on the \emph{ex-post} optimal matching when the market is unbalanced. 

Spatial models are used for studying ride-hailing and ride-sharing applications. In particular, 
\citealt{OzanSaban}, used such models to study spatial price discrimination and 
\citealt{Omar2020} for price optimization and its effect on demand and supply.

Finally, our results are reminiscent of the effects of  the imbalance between two sides of the market in stable marriage problem  with random preference  \citealt{Ashlagi17,pittel2017, kanoria2020random}  and  metric preferences   \citealt{Abadi2017StableMI},  except that we are interested in cost efficiency of matching rather than its stability.  Also, tangential to our result is the work by \citealt{KlempererBulow} on the effect of extra bidders in auctions.

\section{Model and Main Results}

\subsection{Spatial setting and matching algorithms}

  Time passes in discrete periods. Let $R=\{r_1,\ldots, r_n\}$ and $D=\{d_1,\ldots, d_m\}$ be two sets of points chosen independently and uniformly at random on $[0,\ell]$, where $\ell$ can be a function of $n$ and $m$. We call $R$ the set of riders and $D$ the set of drivers. The uniform distribution is not essential for our results; in \Cref{sec.otherdist} we show our main results will not change if the location of drivers and riders is drawn from any distribution with appropriate continuity properties. Throughout the paper we assume $m \geq n$. We say we are in a \textit{balanced market} when $m=n$, and in an \textit{unbalanced market} when $m > n$. 
  
  At $t=0$, platform observes the positions of the drivers. Then, an adversary, which knows all the positions of riders and drivers and the entire state of the matching algorithm, chooses riders one by one and reveals their positions to the algorithm at times $t=1$ to $t=n$. Our main results continue to hold for stochastic arrivals (from a known or unknown distribution).

A matching $M$ is a set of ordered pairs $(r,d)$ such that $r\in R$ and $d\in D$, and each rider or driver belongs to at most one pair in $M$. A matching algorithm takes as input the state of the system at time $t$ (the location of all drivers, as well as riders who have arrived so far) and produces as output a (possibly empty) matching between drivers and riders who are in the market at time $t$. Each unmatched rider incurs a \textit{waiting cost} $c$ per unit of time. Matches are irrevocable.

 The \textit{distance cost} of matching a rider at position $r$ with a driver at position $d$  is equal to $|r-d|$. The cost of a matching algorithm $ALG$, $cost(ALG)$, is equal to the sum over the distance costs of all the pairs plus any waiting cost that riders incur. The goal is to choose a matching to minimize the total cost.

We consider three matching algorithms:
\begin{itemize}
\item The \textit{optimal} matching algorithm, knowing that the order of arrivals is adverserial, it takes as input the current location of all remaining drivers and the location of the newly arrived rider, and outputs a driver to be matched with the new rider. This algorithm is complex.
\item The \textit{omniscient} algorithm is a theoretical benchmark that performs \textit{better} than optimal. It captures the idea that the planner can  engage in acquiring costly information about the location of future riders' arrivals. This algorithm  knows the locations and arrival times of all the riders in advance and matches them optimally when they arrive. 
\item The \textit{greedy} matching algorithm matches each rider to the closest driver.
\end{itemize}

We denote the cost of omniscient and greedy algorithms by $cost(OMN)$ and $cost(Greedy)$, respectively. 

A few remarks about these algorithms worth pointing out. First, for a fixed set of riders and drivers, the omniscient algorithm has a lower cost than the optimal algorithm, which itself has a lower cost than the greedy.

Secondly, while the optimal algorithm is generally complex, the omniscient algorithm is more tractable. In fact, $cost(OMN)$ is the solution to the following integer program:
\begin{align}\label{IP}
\begin{split}
\text{minimize } &\sum_{i=1}^n\sum_{j=1}^m x_{(i,j)}|r_i-d_j| \\
\text{s.t. } & \sum_{j=1}^n x_{(i,j)}= 1, \forall  r_i\in R,\\
& \sum_{i=1}^n x_{(i,j)}\leq 1, \forall d_j\in D,\\
& x_{(i,j)}\in\{0,1\}.
\end{split}
\end{align}

Importantly, the omniscient algorithm is ``detail-free'' in the sense that its performance is independent of waiting cost $c$ and whether the arrival process of riders is random or adversarial. The reason is that the omniscient algorithm knows the optimal \emph{ex post} matching, and as such matches each rider immediately upon arrival to its optimal \emph{ex post} driver; as such, it incurs no waiting cost and performs equally well for any arrival process.

Similarly, the greedy algorithm incurs no waiting cost and thus its performance is independent of $c$. Nevertheless, unlike the omniscient, this algorithm is na\"{\i}ve and makes mistakes both \textit{ex post} and \textit{ex ante}. This is because greedy algorithm ignores that there are \textit{externalities} involved with removing a driver in the system. \autoref{fig: bad example greedy} illustrates a simple example in which greedy makes an \textit{ex post} mistake. Here, black circles are the drivers and $r_1$ arrives before $r_2$. Greedy matches $r_1$ to the driver on the right, whereas the omniscient matches $r_1$ to the driver on the left, making $r_1$ slightly worse off and $r_2$ much better off.

Note that in the example of \autoref{fig: bad example greedy}, a realistic matching algorithm could in principle wait for one period so that $r_2$ arrives, and then choose the same matching as the omniscient. The total distance cost of such ``patient'' algorithm is same as the omniscient, but the algorithm incurs a waiting cost $c$ because rider $r_1$ is forced to wait for one period.\footnote{This is essentially the idea behind the algorithm used by Uber: \url{https://www.uber.com/us/en/marketplace/matching/}}


\begin{figure}[htbp]
\centering
  \begin{tikzpicture}[
    x=5.5mm,
    y=5.5mm,
  ]

    \begin{scope}[
      semithick,
      |-|,
      >={Stealth[]},
    ]
      \draw (-.3, 0) -- (7, 0);
    \end{scope}

    \draw [-,thick, gray] (.5,0) to [out=310,in=230] (6,0);
    \draw [-,thick, gray] (3,0) to [out=270,in=270] (4.5,0);

    \tikzstyle{dot2} = [fill=white, draw];
    \tikzstyle{dot1} = [fill=darkgray, draw];
     \foreach \x / \y in {{.5/1},{3/2},{4.5/1}, {6 /2}} {
        \path[dot\y] (\x, 0) circle[radius=2.5mm];
    }
    \node[](e1) at (3, 0) {$r_1$};
    \node[](e1) at (6, 0) {$r_2$};
  \end{tikzpicture}
\caption{Greedy algorithm can perform much worse than the omniscient algorithm. Here, dark circles are drivers and white nodes are riders. $r_1$ arrives before $r_2$. Greedy matches $r_1$ to the right driver and $r_2$ to the left, which is an \emph{ex post} mistake.}
\label{fig: bad example greedy}
  \end{figure}
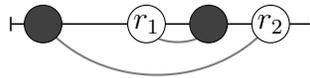


Finally, the greedy algorithm has a `decentralized' interpretation. Suppose there is no central planner, riders arrive at the market one by one, and choose a driver to maximize their own utility. Since we assumed drivers are homogeneous except for their locations, a new rider will pick the closest driver, which is precisely what the greedy algorithm would do. As such, our analysis of the greedy algorithm is in effect analyzing a decentralized setting with selfish agents.

\subsection{Sophisticated matching algorithms and the value of extra drivers}

To quantify the value of raising the supply, we first compare the performance of omniscient matching algorithm in a balanced market to the performance of the omniscient algorithm in an unbalanced market---when there are $n$ riders and $m\geq n(1+\epsilon)$ drivers.

We prove that the cost of the omniscient matching in the balanced market is $\Theta(\ell\sqrt{n})$, whereas if the market is unbalanced then the cost is $O(\ell)$.  More precisely, we prove the following theorem:

\begin{theorem}\label{thm:balancedvsunbalanced}
  Let $R=\{r_1,\ldots,r_n\}$ and $D=\{d_1,\ldots, d_m\}$ be two sets random points drawn independently from a uniform distribution on $[0,\ell]$. Then
\begin{enumerate}
    \item\label{thm:balancedomn} In a balanced market, there exists $C>0$ such that for all $n\geq 1$
    \[\mathbb{E}[cost(OMN)] \geq C\ell\sqrt{n}.\] 
\item \label{thm:unabalancedomn}
In an unbalanced market where $m\geq (1+\epsilon)n$, for some $\epsilon>0$. Then, 
there exists some $C_\epsilon>0 $ such that for all $n\geq 4\epsilon^{-1}$
\[\mathbb{E}[cost(OMN)]\leq C_\epsilon \ell.\]
\end{enumerate}
The expectation in both parts is over the randomness of $R$ and $D$ on the interval $[0,\ell]$.
\end{theorem}

Hence, when we have access to computational power and perfect prediction of future arrivals, unbalancedness can generate stark reduction in costs: the cost in an unbalanced market is independent of $n$, whereas in a balanced market the cost grows at rate $\sqrt{n}$.

\subsection{The greedy matching algorithm and the value of extra drivers}

Omniscient algorithm is practically infeasible, since it can predict the entire future demand and has limited computational power. One may wonder whether increasing supply is highly valuable only for omniscient. Putting differently, how much of the gain from increasing supply can we realistically achieve?  To answer this question, we consider the greedy algorithm.  This algorithm does not exploit (and does not require) information about future arrivals, and it does not take into account the positions of any driver except the nearest; it ignores the global `network' structure of the system.  In the following theorem, we prove that the greedy algorithm with $(1+\epsilon)n$ drivers has an expected matching cost of $O(\ell\log^3 (n))$.

\begin{theorem}\label{thm:unbalancedgreedy}
Consider an unbalanced market where $m\geq (1+\epsilon)n$, for some $\epsilon>0$. Then, 
there exists some $C_\epsilon>0 $ such that for large enough $n$ 
\[\mathbb{E}[cost(Greedy)]\leq C_\epsilon \ell\log^3(n),\]
where the expectation is over the randomness of $R$ and $D$ on the interval $[0,\ell]$.
\end{theorem}

This result shows that the na\"{\i}ve greedy algorithm with $n(1+\epsilon))$ drivers performs \textit{much} better than the omniscient algorithm with $n$ drivers, and not much worse than the omniscient algorithm with $n(1+\epsilon))$ drivers. Importantly, the bound for the greedy algorithm holds \textit{even if riders arrive in an adversarial order.}

\subsection{Matching for other distributions}\label{sec.otherdist}

So far we have assumed that riders and drivers are uniformly distributed, but our results yield corollaries for other distributions.

Suppose that riders and drivers are distributed independently and identically on the unit interval, with cumulative distribution function $F:[0,1] \rightarrow [0,1]$, satisfying $F(0) = 0$ and $F(1) = 1$.  Given a random variable $x$ with CDF $F$, the random variable $F(x)$ is uniformly distributed on $[0,1]$.  Let us define the normalized distance cost of matching rider $r_i$ and driver $d_j$ as $| F(r_i) -  F(d_j)|$, and let $cost_F(\cdot)$ be the total cost of an algorithm when distance is measured with normalized cost.

\begin{corollary}
If $F$ is Lipschitz-continuous, then in a balanced market, there exists $C>0$ such that for all $n\geq 1$   \[\mathbb{E}[cost(OMN)] \geq C\sqrt{n}.\]
\end{corollary}
\begin{proof}
Let $OMN_F$ be the omniscient algorithm that minimizes the \textit{ex post} normalized cost.  By \Cref{thm:balancedvsunbalanced}, there exists $C' > 0$ such that for all $n \geq 1$,
\begin{equation}\label{eq:lips_upper_1}
    C' \sqrt{n} \leq \mathbb{E}[cost_F(OMN_F)] \leq \mathbb{E}[cost_F(OMN)]
\end{equation}
Let $\alpha$ be the Lipschitz constant of $F$.  For any pair $(r_i,d_j)$, we have $|F(r_i) - F(d_j)| \leq \alpha |r_i - d_j|$. Hence,
\begin{equation}\label{eq:lips_upper_2}
    \mathbb{E}[cost_F(OMN)] \leq \alpha \mathbb{E}[cost(OMN)].
\end{equation}
Combining Equations \ref{eq:lips_upper_1} and \ref{eq:lips_upper_2} completes the proof.
\end{proof}

Let the normalized location of rider $r_i$ be $F(r_i)$, and similarly for drivers.  Let $Greedy_F$ be the Greedy algorithm but with the normalized locations as inputs.

\begin{corollary}
If $F^{-1}$ is Lipschitz continuous, then for any an unbalanced market where $m\geq (1+\epsilon)n$, for some $\epsilon>0$, there exists some $C_\epsilon>0 $ such that for large enough $n$ 
\[\mathbb{E}[cost(Greedy_F)]\leq C_\epsilon \log^3(n).\]
\end{corollary}
\begin{proof}
Let $\alpha$ be the Lipschitz constant of $F^{-1}$. For any pair $(r_i,d_j)$, we have
$$|r_i - d_j| = |F^{-1}(F(r_i)) - F^{-1}(F(d_j)| \leq \alpha |F(r_i) - F(d_j)|.$$
Hence, we have that
\begin{equation}\label{eq:lips_lower_1}
\mathbb{E}[cost(Greedy_F)] \leq \alpha \mathbb{E}[cost_F(Greedy_F)]. 
\end{equation}
Moreover, by  \Cref{thm:unbalancedgreedy}, there exists $C'_{\epsilon} > 0$ such that
\begin{equation}\label{eq:lips_lower_2}
    \mathbb{E}[cost_F(Greedy_F)] \leq C'_{\epsilon} \log^3(n).
\end{equation}
Combining Equations \ref{eq:lips_lower_1} and \ref{eq:lips_lower_2} completes the proof.
\end{proof}

\subsection{What if we can ignore some (far away) riders?}

Finally, we investigate whether the difficulty of the balanced-market problem is caused by the requirement that the platform serves all riders, even those who are far away.  We study a relaxed problem, in which the platform can leave a rider unmatched at a penalty $\nu$. We allow $\nu$ to depend on the market size $n$.
 First note that without a lower bound for $\nu$, the problem is trivial.  For instance, if $\nu = \frac{{\ell}}{n}$, the the platform will choose to match no rider, and pay a cost $n\frac{{\ell}}{n}={\ell}$. In fact, for any $\nu < \frac{{\ell}}{n^{1/2}}$ and for large $n$, the platform prefers matching \textit{no} riders to matching \textit{every} rider.  To exclude such trivial cases, we assume instead that $\nu \geq \frac{{\ell}}{n^{1/2 - \delta}}$ for some $\delta > 0$.  
 With this, the cost of the optimal matching is given by the following integer programming:
\begin{align}\label{IP with nu}
\begin{split}
\text{minimize } &\sum_{i=1}^n\sum_{j=1}^m x_{(i,j)}|r_i-d_j|+\sum_{i=1}^n\nu(1-\sum_{j=1}^m x_{(i,j)}) \\
\text{s.t. } & \sum_{j=1}^n x_{(i,j)}\leq 1, \forall  r_i\in R\\
& \sum_{i=1}^n x_{(i,j)}\leq 1, \forall d_j\in D,\\
& x_{(i,j)}\in\{0,1\}.
\end{split}
\end{align}
Here for each rider $r_i$, the value of $(1-\sum_{j=1}^m x_{ij})$ shows whether this rider is matched or not. We call $M^*$, the  matching corresponding to the solution of \eqref{IP with nu}, as \textit{the minimum matching with penalty}. 
 The next theorem then proves that the expected cost of the omniscient matching is $\Omega(n^\delta)$.

\begin{theorem}\label{thm: balanced} 
Given integers $n\geq 0$ and $\ell$,  let $R=\{r_1,\ldots,r_n\}$ and $D=\{d_1,\ldots, d_n\}$ be two sets random points drawn independently from a uniform distribution on $[0,\ell]$. Then for $M^*$,  the minimum cost matching with penalty, 
    \begin{enumerate}
    \item \label{part: balanced thm no cost}If $\nu\geq \ell$ then there exist $C_1, C_2>0$ such that for all $n\geq 1$
    \[C_1\ell\sqrt{n}\geq \mathbb{E}[cost(M^*)]\geq C_2\ell\sqrt{n}.\]
        \item \label{part: balanced thm with cost} If $\nu\geq \frac{\ell}{n^{1/2-\delta}}$ for some $\delta>0$ then there exists a constant $C>0$ such that for large enough $n$ we have 
    \[\mathbb{E}[cost(M^*)]\geq C\ell n^{\delta}.\]
    \end{enumerate}
The expectation in both parts is over the randomness of $R$ and $D$ on the interval $[0,\ell]$. 
\end{theorem}

 Hence, giving the platform the option to forego matching distant riders still leads to costs that are at least polynomial in $n$.

Note that by setting $\nu=\ell$ the cost of the minimum matching with and without penalty are the same, because  the penalty of not matching a rider is more than the cost of matching to some driver within distance $\ell$. 

\subsection{Experiments}

 One may suspect that our theoretical results hold only in the limit as $n$ gets very large, and wonders whether they continue to hold in small markets. To address this issue, we run simulations and compare omniscient and greedy matching algorithms. In particular, we ask: in a unit interval and for a given number of riders, how many extra drivers are needed for the greedy to beat the omniscient algorithm? 
 
  \autoref{fig: beat_optimal} depicts, for a given number of drivers, the percentage of extra drivers required by the greedy to beat the omniscient in the balanced market. As it is clear, the percentage is decreasing as the market gets larger, but it is small even for small markets. For instance, for $n=25$, having 1 excess driver and greedily matching riders has a lower total distance cost than the omniscient algorithm.  For $n = 100$, $4$ excess drivers are enough, and for $n=1000$, and $13$ excess drivers are enough. 
 
\begin{figure}[!htbp]
\centering
\includegraphics[width=7cm]{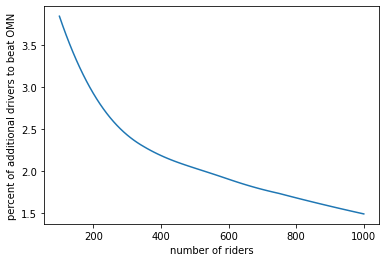}
\caption{Additional percentage  of drivers required for greedy to outperform omniscient.}
\label{fig: beat_optimal}
\end{figure}

\section{The Unbalanced Market with More Drivers}\label{sec: unbalanced}

This section is on the analysis of the unbalanced market. The key tool  is to represent a sequence of drivers and riders with a random walk, as discussed in  \Cref{sec: RW}.
Then in  \Cref{sec: proof unbalanced} we describe the relation of the cost of omniscient algorithm and the random walk to prove  \Cref{thm:balancedvsunbalanced} Part \ref{thm:unabalancedomn}. 
Finally,   \Cref{sec: greedy} is on the analysis of the greedy algorithm and the proof of  \Cref{thm:unbalancedgreedy}.

\subsection{A random walk representation}\label{sec: RW}
By representing the sequence of riders and drivers with a  random walk, we show that it is possible to divide $[0,\ell]$ into  $m-n+1$ (random) sub-intervals, which we refer to as \textit{slices}, such that: 
\begin{enumerate} [i] \item\label{condition i} There are more drivers than riders in each slice. \item\label{condition ii} The total number of riders and drivers in each slice decays exponentially fast. 
\end{enumerate}

Given the sets of riders $R=\{r_1,\ldots,r_n\}$ and drivers $D=\{d_1,\ldots d_m\}$,  consider the walk $W:[0,\ell]\rightarrow \mathbb{Z}$, such that $W(0)=0$ and for each $y>x$,
\begin{equation}\label{eq: walk}
W(y)-W(x)= |D\cap [x,y]| - |R\cap [x,y]|.
\end{equation}
 Define the \textit{exit time} $\gamma_i$, as  the last time the walk hits $i$, i.e.,
\begin{equation}\label{eq: gamma}
\gamma_i=\sup_{t\geq 0}\{t:W(t)-W(0)\leq i \}.
\end{equation}
Observe that $W(\ell)-W(0)=m-n$. So, there are in total $m-n+1$ exit times for each $i\in\{0,1,\ldots, m-n\}$, where $\gamma_{m-n}=\ell$.  
Consider the set of slices $[0,\gamma_{0}],[\gamma_0,\gamma_1],\ldots, [\gamma_{m-n-1},\gamma_{m-n}]$. By definition, in each slice $[\gamma_{i-1},\gamma_{i}]$, the number of drivers is one more than the number of riders, and hence condition \ref{condition i} is satisfied. So, it remains to prove condition \ref{condition ii}.

 We say  $W$ \textit{hops up} at $x$ if $x\in D $ and \textit{hops down} if $x\in R$. 
Let the random variable $H_T$ show the number of hops  in $[0,T]$, 
\begin{equation}\label{eq: definiiton H_T}
    H_T=|\{x\leq T: x\in R\text{ or } x\in D \}|.
\end{equation}
Also, define
 \begin{align}\label{eq: defi hat Hi}
 \begin{array}{lr}
 \hat H_0=H_{\gamma_0},&\\
  \hat H_i=H_{\gamma_i}-H_{\gamma_{i-1}}-1,&\text{ for }0\leq i\leq m-n+1.
\end{array}
\end{align}
For condition \ref{condition ii}, we need to study the distribution of $\hat H_i$s, since  there are $\frac{\hat H_i}{2}$ riders and $\frac{\hat H_i}{2}+1$ drivers in $[\gamma_{i-1},\gamma_i]$.

Consider a random walk that independently at each step hops up  with probability $\frac{m}{n+m}=\frac{1+\epsilon}{2+\epsilon}$ or  down with probability $\frac{1}{2+\epsilon}$. For such a walk, it is known that the number of hops between two consecutive exit times are independent and their expected values are bounded by a constant \citealt{bertoin_doney_1996, Janson86}. 
In our case,  eventhough the probability of hopping up and down changes over time for the walk $W$, it is still possible to find the distribution of $\hat H_i$s and upper bound them.

\begin{lemma}\label{lm: dist hat H_i}
Let $R$ and $D$ be as in  \Cref{thm:balancedvsunbalanced}, and let $\hat H_i$ be defined as in \eqref{eq: defi hat Hi}. Then 
\begin{enumerate}
\item\label{part: dist H0} For $i=0$
\[\mathbb{P}(\hat H_0=2k)=\frac{m-n}{m+n-2k}\frac{{2k\choose k}{m+n-2k\choose n-k}}{{m+n\choose n}}.\]
    \item\label{part: dist Hi} For $i\geq 1$
\[\mathbb{P}(\hat H_i=2k)=\frac{{2k\choose k}}{k+1}\frac{{m+n-2k\choose n-k}}{{m+n\choose n}},\]
\item\label{part: bound dist Hi} If $m\geq (1+\epsilon)n$, and $n\geq 4\epsilon^{-1}$ then for small enough $\epsilon>0$, there exists some $C>0$ such that for all $i\geq 1$ and $k\leq n$,
\[\mathbb{P}(\hat H_i=2k)\leq \frac{C}{k\sqrt k }(1- \frac{\epsilon^2}{18})^k,\]
and for $i=0$,
\[\mathbb{P}(\hat H_0=2k)\leq \frac{C}{\sqrt k }(1- \frac{\epsilon^2}{18})^k.\]
\end{enumerate}
\end{lemma}
\begin{proof} 
We begin by proving Part \ref{part: dist H0}.
The first step is to represents the hops on the walk $W$ with a lattice path, $P$, from $(0,0)$ to $(n+m,m-n)$.  Starting from $P=(0,0)$ on the lattice and  $x=0$ on the interval $[0,\ell]$, keep increasing $x$ and updating $P$ on the lattice as follows. Whenever there is a hop up ($x\in D$) add $(1,1)$ to the current position $P$, and whenever there is a hop down ($x\in R$) add $(1,-1)$ to $P$ (see  \autoref{fig: W} and \autoref{fig: lattice path}). Note that if there are $h_u$  hop ups and $h_d$ hop downs   in $[0,h]$ then when $x=h$ the lattice path is at
$P=(h_u+h_d,h_u-h_d)$. In particular,
at the end of the process  $x=\ell$ and $P=(n+m,m-n)$. 


\begin{figure}[!tbp]
\centering
    \begin{subfigure}[b]{0.4\textwidth}
  \begin{tikzpicture}[
    x=5.5mm,
    y=5.5mm,
  ]

    \draw[ultra thin,white]
      \foreach \x in {0, ..., 10} {
        (\x, -2) -- (\x, 6)
      }
      \foreach \y in {-2, ..., 5} {
        (0, \y) -- (11, \y)
      }
    ;  

    \begin{scope}[
      semithick,
      ->,
      >={Stealth[]},
    ]
      \draw (0, 0) -- (11.5, 0);
    \end{scope}

    Random walk

        \draw [ thick] (0,0) -| (.5,1) -| (2,0) -| (2.8,-1) -| (4,0) -| (4.8,1)-| (5.5,2)-| (7,3)-| (7.7,4)-| (9,3)-| (10,2)-|  (11,2);

    \tikzstyle{dot2} = [fill=white, draw];
    \tikzstyle{dot1} = [fill=darkgray, draw];
     \foreach \x / \y in {{.5/1},{2/2}, {2.8 /2}, {4/1}, {4.8/1}, {5.5/1}, {7/1}, {7.7/1}, {9/2},{10/2}} {
        \path[dot\y] (\x, 0) circle[radius=1.3mm];
    }
    
	\node[](e1) at (2, -1.5) {$\tau_0$};
\draw[->] (2,-.3) -- +(e1);
	\node[](e1) at (4.8, -1.5) {$\gamma_0$};
\draw[->] (4.8,-.3) -- +(e1);
  \end{tikzpicture}
    \caption{ The walk $W$ hops up with any driver (dark node) and  hops down with any rider (white node).  Also, compare the last exit time $\gamma_0$ with the first return time $\tau_0$ (see \eqref{eq: gamma} and \eqref{eq: defi tau}). }
  \label{fig: tau vs gamma}
    \label{fig: W}
  \end{subfigure}
  \hspace{1.cm}
  \begin{subfigure}[b]{0.4\textwidth}
  \begin{tikzpicture}[
    x=5.5mm,
    y=5.5mm,
  ]

    \draw[ultra thin]
      \foreach \x in {0, ..., 10} {
        (\x, -2) -- (\x, 6)
      }
      \foreach \y in {-1, ..., 5} {
        (0, \y) -- (11, \y)
      }
    ;

    ;  
    \begin{scope}[
      semithick,
      ->,
      >={Stealth[]},
    ]
      \draw (0, 0) -- (11.5, 0);
    \end{scope}

    Random walk
    \draw[thick]
      (0, 0)
      \foreach \x in {1,-1,-1,1,1,1,1,1,-1,-1,1} {
        -- ++(1,\x)
      }
    ;  

    \tikzstyle{dot2} = [fill=white, draw];
    \tikzstyle{dot1} = [fill=darkgray, draw];
     \foreach \x / \y in {{1/1},{2/2}, {3 /2}, {4/1}, {5/1}, {6/1}, {7/1}, {8/1}, {9/2},{10/2}} {
        \path[dot\y] (\x, 0) circle[radius=1.3mm];
    }

  \end{tikzpicture}
    \caption{The lattice path $P$  moves $\nearrow$ with any driver  (dark node) and it moves $\searrow$ with any rider (white node). Note the lattice path is identified with ordering of riders and drivers and not their exact location.
    }
  \label{fig: lattice path}
  \end{subfigure}
  \end{figure}
By counting the corresponding lattice paths, we find the distribution of  $\hat H_0$. Observe that there are equal number of drivers and riders in $[0,\gamma_0]$. So, if $\hat H_0=2k$, then  the first part of $P$ is from $(0,0)$ to $(2k,0)$, and as a result, it can be formed in ${2k\choose k}$ different ways. Since for all $x\geq \gamma_0$ we have $W(x)-W(\gamma_0)>0$,
 the second part of the lattice path, which is from $(2k,0)$ to $(m+n,m-n)$, should not touch level zero, where by level $i$ we mean the line $y=i$ on the plane.
By  \citealt{catalan_general},  the number of such lattice paths from $(0,0)$ to $(a+b,a-b)$ (with $a\geq b$)  is 
\begin{equation}\label{eq: general catalan}
    \frac{a-b}{a+b}{a+b\choose b}.
\end{equation}
In our case $a=m-k$ and $b=n-k$, and there are  $\frac{m-n}{m+n-2k}{m+n-2k\choose n-k}$ ways to complete the path from $(2k,0)$ to $(m+n,m-n)$ such that it never touches level zero.  Since the total number of paths is ${m+n\choose n}$,  we get
\[\mathbb{P}(\hat H_0=2k)=\frac{m-n}{m+n-2k}{m+n-2k\choose n-k}\frac{{2k\choose k}}{{m+n\choose n}}.\]

Likewise, we use a path counting argument to find the probability of the event $\hat H_i=2k$  for Part \ref{part: dist Hi}. Divide the lattice path $P$ into three parts corresponding to the intervals $[0,\gamma_{i-1}]$, $[\gamma_{i-1},\gamma_{i}]$ and $[\gamma_i,\ell]$.
The first path, $P_0$, starts at $(0,0)$ and ends at $o_{i-1}=(H_{\gamma_{i-1}},i-1)$. The reason is that the interval $[0,\gamma_{i-1}]$ contains $H_{\gamma_{i-1}}$ total hops,  and
by the exit time $\gamma_{i-1}$  we see $i-1$ more hop ups than hop downs.
Recall that if $\hat H_i=2k$, then there are $k+1$ hop ups and $k$ hop downs in $[\gamma_{i-1},\gamma_{i}]$. 
Therefore, the second path, $P_1$, starts at $o_{i-1}$ and ends at $o_{i}=o_{i-1}+(2k+1,1)$.  Finally,  $P_2$ is from $o_{i}$ to $(m+n,m-n)$.
Since $\gamma_i$ is the $i^{th}$ exit time, the lattice path $P_1$ never crosses  level $i$. So, by plugging in $a=k+1$ and $b=k$ in \eqref{eq: general catalan}, we see that the number of such paths from $o_{i-1}$ to $o_{i}$ is $\frac{{2k+1\choose k}}{k+1}$. 

Next, we contract $P_1$ and count the number of paths $P_0$ and $P_2$  together. Let $P.(o_{i-1},o_{i})$ be the lattice path  formed by  removing $P_1$  from $P$, and instead starting the path $P_2$ at the point $o_{i-1}$.
 Then $P.(o_{i-1},o_{i})$ is a path from $o=(0,0)$ to $f=(m+n-2k-1,m-n-1)$. 
Now, given any path $Q$ from $o$ to $f$, we claim that there is a unique way to insert the path $P_1$ back into $Q$ such that $P_1$ corresponds to the hops between $(i-1)^{th}$ and $i^{th}$ exit times. In fact,  find the last intersection of $Q$ with  level $i-1$, and let $Q_0$ and $Q_1$ be the part of $Q$ before and after the intersection, respectively.  Since we want $P_1$ to correspond to the hops between $(i-1)^{th}$ and $i^{th}$ exit times, it needs to appear right after $Q_1$. So, add back $P_1$ at the intersection, i.e., consider the path $P'$ that starts with $Q_0$ and then continues with $P_1$ and $Q_1$, in this order. 
Recall that $Q_1$ starts at level $i-1$ in $Q$ and it never crosses level $i-1$. Therefore, the part of $P'$ after $P_1$, never crosses level $i$. 
Hence, $P_1$ corresponds to the hops between levels $i-1$ and $i$ in $P'$. 

As a result of the above arguments, if we fix $P_1$, then there is a one to one correspondence between paths  from $(0,0)$ to $(m+n-2k-1,m-n-1)$ and paths from $(0,0)$ to $(n+m,m-n)$ such that $P_1$ appears between levels $(i-1)^{th}$ and $i^{th}$. Therefore,
\[ \mathbb{P}(H_i=2k)=\frac{{2k\choose k}}{k+1}\frac{{m+n-2k-1\choose n-k}}{{m+n\choose n}}.\]

It remains to prove Part \ref{part: bound dist Hi}.  By Stirling's approximation, there exists constants $C>C'>0$ such that for all $k$, we have $C' \frac{2^{2k}}{\sqrt k}\leq {2k\choose k}\leq C \frac{2^{2k}}{\sqrt k}$.
So, we bound the following
    \begin{align*}
   2^{2k}\frac{{m+n-2k\choose n-k}}{{m+n\choose n}}
    &=\frac{m-k}{m+n-2k}\prod_{i=1}^k(1-\frac{m-n-1}{m+n-2k+2i-1})(1+\frac{m-n}{m+n-2k+2i})\\
    &\leq\prod_{i=1}^k(1-\frac{(m-n-1)(m-n)}{2(m+n-2k+2i-1)(m+n-2k+2i)}),
    \end{align*}
where in the last inequality we used the fact that $(m-n-1)^2\geq 4(n-k+i)$. Recall that $m-n\geq \epsilon n$. Then
    for any $i$  and $\epsilon<1$, \[\frac{m-n-1}{m+n-2k+2i-1}\frac{m-n}{m+n-2k+2i}\geq\frac{(m-n)^2}{(m+n)^2}\geq\frac{\epsilon^2}{18}.\] 
    Combining this with the previous inequalities, 
    \begin{equation}\label{eq: upper bound with path counting}
         2^{2k}\frac{{m+n-2k\choose n-k}}{{m+n\choose n}}\leq (1- \frac{\epsilon^2}{18})^k. 
    \end{equation}
Therefore, by Part \ref{part: dist Hi}, for $i\geq 1$,
\[\mathbb{P}(H_i=2k)=\frac{{2k\choose k}}{k+1}\frac{{m+n-2k\choose n-k}}{{m+n\choose n}}\leq \frac{C}{\sqrt k (k+1)}(1- \frac{\epsilon^2}{18})^k,\]
and by Part \ref{part: dist H0}, for $i=0$,
\[\mathbb{P}(H_0=2k)={2k\choose k}\frac{m-n}{m+n-2k}\frac{{m+n-2k\choose n-k}}{{m+n\choose n}}\leq \frac{C}{\sqrt k }(1- \frac{\epsilon^2}{18})^k.\]
\end{proof}

We would like to point out that a weaker version of the above lemma could be proved without path counting, and instead using Chernoff bound and the negative association of a random permutation of hop ups and hop downs. 
However, the Chernoff bound argument leads to a poly-logarithmic bound (in terms of $n$) and will not  allow us  to give a constant upper bound  on the cost of omniscient.
In  \Cref{sec: proof unbalanced}, we  directly use the distribution of $\hat H_i$ as stated to give a constant bound.

Next, as a corollary of  \Cref{lm: dist hat H_i}, we show that  the number of hops in each slice is with high probability $O(\log(n))$. 
In Section \Cref{sec: greedy},  we will  use this observation  to bound the cost of the greedy algorithm.
\begin{corollary}\label{cor: max Hi}
Let $R$ and $D$ be as in  \Cref{thm:balancedvsunbalanced}, where for some $\epsilon>0$ we have $m\geq (1+\epsilon)n$. Recall the definition of $\hat H_i$ from \eqref{eq: defi hat Hi}. Then there exists some $\alpha>0$ such that
\[\mathbb{P}(\max_{0\leq i\leq m-n+1}\hat H_i\geq \alpha\log n)\leq \frac{1}{n}.\]
\end{corollary}
\begin{proof}
By Part \ref{part: bound dist Hi} of  \Cref{lm: dist hat H_i} there exists some $C>0$ such that for any $i\geq 0$
\[   \mathbb{P}(\hat H_i\geq \alpha\log n)
    \leq  \sum_{k\geq \alpha\log n}\frac{C}{\sqrt k} (1-\frac{\epsilon^2}{18})^{k}.\]
Let $\alpha=\lceil\frac{-3}{\ln(1-\epsilon^2/18)}\rceil$,  then
\begin{align*}
    \mathbb{P}(\hat H_i\geq \alpha\log n)
   & \leq  \sum_{k\geq \alpha\log n}\frac{C}{\sqrt k} e^{k\ln(1-\epsilon^2/18)}\\
   &\leq \sum_{k\geq \alpha\log n}\frac{1}{\sqrt k} e^{-3\ln(n)}\leq  \sum_{k\geq \alpha\log n} n^{-3}\leq n^{-2}.
\end{align*}
Now by a union bound we get the result.
\end{proof}
 
Finally, we want to point out that the walk representation  has been  studied before by  \citealt{holroyd2009geometric,holroyd2020minimal} for matching Poisson point processes in metric spaces, however, their result is on Poisson processes with the same intensity (equivalent to balanced markets). Our belief is that this technique has its novelty that has not been used to a great extent in the literature.


\subsection{Analysis of the omniscient algorithm}\label{sec: proof unbalanced}
In order to bound the cost of omniscient, first, recall the definition of $\gamma_i$ from \eqref{eq: gamma} and the fact that each slice $[\gamma_i,\gamma_{i+1}]$ contains more drivers than riders. A simple upper bound on the cost of omniscient can be obtained by pairing all riders and driver within the same slice.
Define
  \begin{align}\label{eq: defi hat tau}
 \begin{array}{lr}
\hat\tau_0=\gamma_0,&\\
 \hat\tau_i=\gamma_i-\gamma_{i-1},&\text{ for }0\leq i\leq m-n+1.
\end{array}
\end{align}
Then the cost of matching riders and drivers in the same slice is at most $\sum_{i=0}^{m-n+1}\hat H_i\hat \tau_i$, which we will prove in the next result is bounded by a constant factor of the length of the interval.

\begin{lemma}\label{lm: H tau}
Let $R$ and $D$ be as in  \Cref{thm:balancedvsunbalanced}.
Let $\hat\tau_i$  and $\hat H_i$ be defined as in \eqref{eq: defi hat tau} and \eqref{eq: defi hat Hi}. If $m\geq(1+\epsilon)n$ for some $\epsilon>0$, then  there exists $C_\epsilon>0$ such that 
\[\mathbb{E}\Big[\sum_{i=0}^{m-n+1}\hat H_i\hat\tau_i\Big]\leq C_{\epsilon}\ell.\]
\end{lemma}
\begin{proof}

Given $j\geq i$, it is a well-known fact that the distance between $i^{th}$ and $j^{th}$ order statistics of $n+m$ independent uniform samples  on $[0,1]$ is distributed as $Beta(j-i,n+m+1-(j-i))$ with mean $\frac{j-i}{m+n+1}$ (e.g., see \citealt{pitman1999probability}).
Note that $\hat\tau_i$ is the distance between $H_{\gamma_{i-1}}^{th}$ and $H_{\gamma_i}^{th}$ order statistics.
 Since the length of a segment between consecutive hops does not depend on whether it is a hop up or a hop down, we can condition on the number of hops in $[\gamma_{i-1},\gamma_{i}]$ to find the expected length of $\hat\tau_i$. Therefore,
\[\mathbb{E}[\hat\tau_i|\hat H_{i}=2k]=\mathbb{E}[\hat\tau_i|H_{\gamma_i}-H_{\gamma_{i-1}}=2k+1]=\frac{(2k+1)\ell}{m+n+1},\]
where the factor $\ell$ comes from scaling the interval $[0,1]$ to $[0,\ell]$.

With this observation and Part \ref{part: bound dist Hi} of  \Cref{lm: dist hat H_i}, for $i\geq 1$ we get,
\begin{align*}
    \mathbb{E}[\hat H_i\hat\tau_i]&=\sum_{k=0}^n2k \mathbb{P}(\hat H_i=2k) \mathbb{E}[\tau_i|\hat H_i=2k]\\
    &=\sum_{k=0}^n\frac{2k(2k+1)\ell}{n+m+1}\mathbb{P}(\hat H_i=2k)\\
    &\leq\frac{C\ell}{n+m+1}\sum_{k=0}^n\sqrt{k}(1-\frac{\epsilon^2}{18})^k.
\end{align*}
 Similarly, for $i=0$,
        \begin{align*}
         \mathbb{E}[\hat H_0\hat\tau_0]
    &\leq\frac{C\ell}{n+m+1}\sum_{k=0}^nk(1-\frac{\epsilon^2}{18})^k.
    \end{align*}
For any $q\in(0,1)$, it is known that $\sum_{k=1}^\infty k(1-q)^k=\frac{1-q}{q^2}$. Therefore,  for some $C'>0$ and all $i\geq 0$ we get
    \begin{align*}
         \mathbb{E}[\hat H_i\hat\tau_i]&\leq\frac{C\ell}{n+m+1}\sum_{k=0}^nk(1-\frac{\epsilon^2}{18})^k\\
         &\leq \frac{\epsilon^{-4}C'\ell}{n+m+1}.
    \end{align*}
    As a result,
    \[\sum_{i=0}^{m-n+1}\mathbb{E}[\hat H_i\hat\tau_i]\leq C'\epsilon^{-4}\ell.\]
\end{proof}

Now, we are ready to finish the proof of Part \ref{thm:unabalancedomn} of  \Cref{thm:balancedvsunbalanced}.
\begin{proof}[Proof of   \Cref{thm:balancedvsunbalanced}, Part \ref{thm:unabalancedomn}]
Given an instance of $R$ and $D$, construct the walk $W$ as in \eqref{eq: walk}. Recall the definition of $\gamma_i$ from \eqref{eq: gamma} and slices $[0,\gamma_0],[\gamma_0,\gamma_1],\ldots,[\gamma_{m-n-1},\gamma_{m-n}]$. 
By definition,  there are $1+\hat H_i/2$ drivers and $\hat H_i/2$ riders  in the slice $[\gamma_{i-1},\gamma_{i}]$ for $i\geq 1$. Also, there are equal number of riders and drivers in $[0,\gamma_0]$. 
Therefore, there exists a matching $M_\gamma$ such that for each pair $(r,d)\in M_\gamma$ both $r$ and $d$ appear in the same slice. Then
\[cost(M_\gamma)\leq \sum_{i=0}^{m-n} \hat H_i(\gamma_i-\gamma_{i-1}).\]
As a result, if $OMN$ is the minimum cost matching,
\begin{align*}
    \mathbb{E}[cost(OMN)]&\leq\mathbb{E} [cost(M_\gamma)]\\
    &\leq \mathbb{E}[\sum_{i=0}^{m-n} \hat H_i(\gamma_i-\gamma_{i-1})].
\end{align*}
Now,  by  \Cref{lm: H tau} we get the desired result.
\end{proof}

\subsection{Analysis of greedy}\label{sec: greedy}
We continue with studying the unbalanced market, where the position of riders are revealed in an online fashion. We formalize the online matching model and then  give an upper bound on the cost of the greedy algorithm.

In the adversarial online matching model, initially, the algorithm has access to the set of drivers $D$. At each step, an adversary, with full information of the state of the algorithm, chooses a rider $r\in R$ and passes its location to the algorithm. The algorithm has one chance to match $r$, irrevocably. 
The most straightforward way to match riders, as stated in Algorithm \ref{alg: greedy}, is greedy: match the arriving rider to the closest unmatched driver.

\begin{algorithm}
	\SetAlgoNoLine
\KwIn{ $D=\{d_1,\ldots, d_m\}$.
}
$M=\emptyset$

\While{a new rider $r$ arrives}{
     $d^*$= the closest unmatched driver to $r$

     $M=M\cup(r,d^*)$
  }

 return $M$
 \caption{Greedy}
 \label{alg: greedy}
\end{algorithm}

The rest of the section is on the proof of  \Cref{thm:unbalancedgreedy}. Our proof is based on showing a stronger result:  with high probability, greedy 
matches each rider to one of the $O(\log^2(n))$ closest drivers. 
The following observation shows that such a driver will not be too far.

\begin{proposition}\label{prop: max dist}
Given a set of uniform random points $d_{(1)}\leq d_{(2)}\leq \ldots \leq d_{(m)}$ on $[0,\ell]$, then
\[\mathbb{P}(\max_{1\leq i\leq m}(d_{(i)}-d_{(i-1)})\geq 2\frac{\ell\log(m)}{m})\leq \frac{1}{m}.\]
\end{proposition}
\begin{proof}
We prove it for $\ell=1$, then scaling all the points by $\ell$ gives the desired statement. 
For simplicity, let $d_{(0)}=0$ and $d_{(m+1)}=1$. For order statistics of $m$ standard uniform random points, it is known that $d_{(i)}-d_{(i-1)}\sim Beta(1,m)$ (see e.g., \citealt{pitman1999probability}). Then for each $m+1\geq i\geq 1$
\[\mathbb{P}(d_{(i)}-d_{(i-1)}\geq t)\leq (1-t)^m.\]
Therefore, by a union bound
\[\mathbb{P}(\max_i(d_{(i)}-d_{(i-1)})\geq 2\frac{\log(m)}{m})\leq m (1-2\frac{\log(m)}{m})^m\leq \frac{1}{m}.\]
\end{proof}

\begin{proof}[Proof of  \Cref{thm:unbalancedgreedy}] 
Recall the definition of $\gamma_i$ from \eqref{eq: gamma}, which partition the interval $[0,\ell]$ into slices  $\Gamma=\{[0,\gamma_0],[\gamma_0,\gamma_1],\ldots,[\gamma_{m-n-1},\gamma_{m-n}]\}$. The main idea of the proof is to show that any rider is matched to a driver within $O(\log n)$ slices away from it.

 By  \Cref{cor: max Hi}, we know there exists some $\alpha>0$ such that 
$\mathbb{P}(\max_{0\leq i\leq m-n+1}\hat H_i\geq \alpha\log n)\leq \frac{1}{n},$
where $\hat H_i$ was defined in \eqref{eq: defi hat Hi}.
Define the event
\[\mathcal{A}_1=\{R\cup D: \max_{0\leq i\leq m-n+1}\hat H_i\leq \alpha\log n \},\]
which indicates that each slice contains at most $\alpha\log(n)$ drivers. Also, let
\[\mathcal{A}_2=\{ D: \max_{1\leq i\leq m}(d_{(i)}-d_{(i-1)})\leq \frac{2\ell\log(n)}{n}n \},\]
where $d_{(1)}\leq \ldots \leq d_{(m)}$ is the order statistic of the set of drivers $D$. 
Let $M_G$ be the matching returned by greedy.
We claim that it is sufficient to only consider instances of $R\cup D$ satisfying both $\mathcal{A}_1$ and $\mathcal{A}_2$, and prove that
\begin{equation}\label{eq: int-step-greedy}
    |\{S\in \Gamma: S\cap I_{(r,d)}\neq \emptyset\}|\leq 2\alpha\log(n)+4, \quad \text{ for each pair $(r,d)\in M_G$,}
\end{equation} 
where $I_{(r,d)}=\big(\min(r,d),\max(r,d)\big)$ is the open sub-interval covered by the  pair $(r,d)$.

To see why proving \eqref{eq: int-step-greedy} is enough, note that  $\mathcal{A}_1$  and  \eqref{eq: int-step-greedy} imply that each rider $r\in R$ is matched to  one of the $\alpha \log(n)(2\alpha\log(n)+4)$ closest drivers to it. Also, since $\mathcal{A}_2$ holds,  we get the following bound for any matched pair $(r,d)\in M_G$,
\[|d-r|\leq 4\alpha\ell(\alpha \log_n+2) \frac{\log^2(n)}{n}.\]
Therefore, when $\mathcal{A}_1$ and $\mathcal{A}_2$ hold, there exists a constant $\alpha'>0$ such that the size of the matching is bounded by $\alpha'\ell\log^3(n)$.
Now,  by  \Cref{prop: max dist} and  \Cref{cor: max Hi}, we know   both events  $\mathcal{A}_1$ and $\mathcal{A}_2$ happen with probability at least $1-\frac{2}{n}$. Using this and the trivial upper bound of $\ell n$ on the size of matching, we get 
\begin{align*}
    \mathbb{E}[cost(M_G)]&\leq \mathbb{E}[cost(M_G)\mid \mathcal{A}_1,\mathcal{A}_2]\mathbb{P}(\mathcal{A}_1,\mathcal{A}_2)+\ell n(1-\mathbb{P}(\mathcal{A}_1,\mathcal{A}_2))\\
    &\leq \alpha'\ell\log^3(n)+2\ell, 
\end{align*}
which proves the statement of the theorem.

It remains to prove \eqref{eq: int-step-greedy}, when both events $\mathcal{A}_1$ and $\mathcal{A}_2$ hold.
Start with $m+1$ components $C=\{[0,d_{(1)}], [d_{(1)},d_{(2)}], \ldots, [d_{(m)},\ell]\}$. 
Consider the following procedure to update $C$ with a run of the algorithm.
As greedy matches a rider $r$ to a driver $d$, update $C$ by merging all the components that have an overlap  with $I_{(r,d)}$. Note that  greedy does not leave an  unmatched driver in $I_{(r,d)}$. So, at each step we know that in each component of $C$ all the drivers except possibly the left-most and the right-most ones are matched to some rider.
Let $C_f$ be the final set after greedy stops (see \autoref{fig: intervals greedy}). 

Note that for any pair $(r,d)\in M_G$, $I_{(r,d)}$ is a subset of some component in $C_f$. So in order to prove \eqref{eq: int-step-greedy},  it is enough to show that each component of $C_f$ has overlap with at most $2\alpha\log(n)+4$ slices. 

Let $K=\lceil2\alpha\log(n)+5\rceil$.
Assume to the contrary that there exists $I\in C_f$ and some $i\geq 0$ such  that $I$ overlaps with $K$ slices
$[\gamma_{i},\gamma_{i+1}],[\gamma_{i+1},\gamma_{i+2}],
\ldots, [\gamma_{i+K-1},\gamma_{i+K}]$. As we noted earlier, all components in $C_f$ must have at most two unmatched drivers. But we show this is not possible for $I$. 
By construction of $C_f$, the component $I$  contains the middle slices $[\gamma_{i+j},\gamma_{i+j+1}]$ for all $1\leq j\leq K-2$,
and it may partially overlaps with the left-most and the right-most slices, $[\gamma_{i},\gamma_{i+1}]$ and $[\gamma_{i+K-1},\gamma_{i+K}]$ without containing them.
By the definition of $\gamma_j$, the number of drivers  in each slice is one more than the number of riders (except possibly the first slice). So, in total, there are at least $K-2$ additional drivers in the middle slices of $I$. Since $\mathcal{A}_1$ holds, there are at most $2\alpha\log(n)$ riders in the left-most and the right-most slices. So, 
\begin{align*}
    |\{D\cap I\}|-|\{R\cap I\}|&\geq 
    -|\{R\cap[\gamma_{i},\gamma_{i+1}]\}|\\
    &\qquad+\sum_{j=1}^K  \Big(|\{D\cap[\gamma_{i+j},\gamma_{i+j+1}]\}|-|\{R\cap[\gamma_{i+j},\gamma_{i+j+1}]\}|
    \Big)\\
    &\qquad\qquad\quad-|\{R\cap[\gamma_{i+K},\gamma_{i+K+1}]\}|\\
    &\geq -2\alpha\log(n)+K\\
    &\geq 3,
\end{align*}
which implies that there must be at least 3 unmatched drivers in $I$, a contradiction. 
Therefore, each $I\in C_f$ has overlap with  at most $2\alpha\log(n)+4$ slices, and our claim \eqref{eq: int-step-greedy} is proved.

\end{proof}

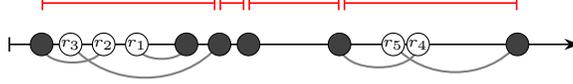
\begin{figure}[!tbp]
\centering
  \begin{tikzpicture}[
    x=5.5mm,
    y=5.5mm,
  ]

    \begin{scope}[
      semithick,
      |->,
      >={Stealth[]},
    ]
      \draw (-.3, 0) -- (13.5, 0);
    \end{scope}

    \begin{scope}[
      semithick, red,
      |-|,
      >={Stealth[]},
    ]
      \draw (.5, 1) -- (4.7, 1);
      \draw (4.8, 1) -- (5.4, 1);
      \draw (5.5, 1) -- (7.7, 1);
      \draw (7.8, 1) -- (12, 1);
    \end{scope}

    \draw [-,thick, gray] (.5,0) to [out=270,in=270] (2,0);
    \draw [-,thick, gray] (2.8,0) to [out=270,in=270] (4,0);
    \draw [-,thick, gray] (1.2,0) to [out=310,in=230] (4.8,0);
    
    \draw [-,thick, gray] (9,0) to [out=310,in=230] (12,0);
    \draw [-,thick, gray] (7.7,0) to [out=270,in=270] (9.6,0);

    \tikzstyle{dot2} = [fill=white, draw];
    \tikzstyle{dot1} = [fill=darkgray, draw];
     \foreach \x / \y in {{.5/1},{1.2/2},{2/2}, {2.8 /2}, {4/1}, {4.8/1}, {5.5/1},{7.7/1}, {9/2},{9.6/2},{12/1}} {
        \path[dot\y] (\x, 0) circle[radius=1.5mm];
    }
    \node[](e1) at (1.2, 0) {\tiny $r_3$};
    \node[](e1) at (2, 0) {\tiny $r_2$};
    \node[](e1) at (2.8, 0) {\tiny $r_1$};
    
    \node[](e1) at (9, 0) {\tiny $r_5$};
    \node[](e1) at (9.6, 0) {\tiny $r_4$};
  \end{tikzpicture}
\caption{An instance of riders (white nodes) and drivers (dark nodes), where the riders arrive in the order $r_1, r_2,\ldots, r_5$. Compare the matching returned by greedy (curved lines) and the components of $C_f$ (red intervals), as constructed in the proof of  \Cref{thm:unbalancedgreedy}.
}
\label{fig: intervals greedy}
  \end{figure}

\section{The Balanced Market}\label{sec: balanced}

We proceed to analyze the balanced market, by taking a closer look on the combinatorial structure of the spatial  matchings in  \Cref{sec: structure matching}. We show that there exists an optimal matching such that the interval between any matched pairs $r$ and $d$ contains equal number of riders and drivers. 
As a consequence of this observation, lower  bounding on the cost of a pair reduces to  lower bounding the first return time of the walk $W$, which was defined in \eqref{eq: walk}.
 \Cref{sec: RW balanced} studies this idea and analyzes the walk when $m-n$ is sub-linear in $n$. Combining the results,  \Cref{sec: balanced proof} is on the proof of  \Cref{thm: balanced}.

\subsection{Matching structure}\label{sec: structure matching}
Given a matching $M$, for any pair $(r,d)$ in $M$, let $I_{(r,d)}=\big(\min(r,d),\max(r,d)\big)$ be the interval covered by the pair. We say $e_1$ is oriented to the \textit{right} if $r_1<d_1$ and oriented to the \textit{left} otherwise.
Two pairs $e_1=(r_1,d_1)$ and $e_2=(r_2,d_2)$ overlap if $I_{e_1}\cap I_{e_2}\neq \emptyset$.
The following observation states that in any optimal matching all the overlapping pairs are oriented in the same direction.
 \begin{proposition}\label{prop: matching direction}
 Given sets of points $R$ and $D$ on $[0,\ell]$, let  $M$ be a minimum matching with penalties on $R$ and $D$ (a solutions of \eqref{IP with nu}).
 Let $e_1,e_2\in M$ be two overlapping pairs. Then $e_1$ is oriented to the right if and only if $e_2$ is oriented to the right.
 \end{proposition}
\begin{proof} 
Assume to the contrary that, there are two overlapping pairs $e_1=(r_1,d_1)$ and $e_2=(r_2,d_2)$ in $M$ such that $e_1$ is oriented to the right and $e_2$ is oriented to the left. It is easy to check that, deleting $e_1$ and $e_2$ and adding the pairs $(r_1,d_2)$ and $(r_2,d_1)$ to $M$ results in a matching with a smaller cost.
\end{proof}

Following the terminology of 
\citealt{holroyd2020minimal}, the overlapping pairs, $e_1$ and $e_2$, are called \textit{nested} if either $I_{e_1}\subset I_{e_2}$ or $I_{e_2}\subset I_{e_1}$. Otherwise, we call the overlapping pairs \textit{entwined}.
A matching $M$ is nested, if all its overlapping pairs are nested. 
The following result shows that it is possible to make any matching nested without changing the cost. We will use this result in the proof of  \Cref{thm: balanced}, to make the optimal matching with penalties a nested optimal matching.

\begin{proposition}\label{prop: nested} 
Given $R$ and $D$, the position of riders and drivers in $[0,\ell]$, let $M$ be any matching between $R$ and $D$. Then there exists a nested matching $M'$ such that 
 $cost_{R,D}(M)=cost_{R,D}( M')$.
\end{proposition}
\begin{proof}
Given a matching $M$, we construct a nested matching from $M$ with the same cost sequentially. Let $S_M$ be the set of the pairs in $M$ which are entwined with at least another pair, i.e.,
\[S_M=\{e\in M : \exists e'\in M \text{ such that $e$ and $e'$ are entwined.}\}\]
Next, we describe a procedure that at each step either the size of $S_M$ decreases, or its size stays the same but the length of the largest interval corresponding to a pair in $S_M$ increases.

Let $e_1$ be a pair with the largest $|I_{e_1}|$ among all the pairs in $S_M$.   Since $e_1\in S_M$, there exists a pair $e_2=(r_2,d_2)$ such that $e_1$ and $e_2$ are entwined. 
By  \Cref{prop: matching direction}, we know that $e_1$ and $e_2$ must be in the same direction. Without loss of generality, assume that they are oriented to the right and also that $r_1\leq r_2$. As a result, $r_1\leq r_2\leq d_1\leq d_2$, and
\[cost(\{e_1,e_2\})=  |d_2-r_2|+|d_1-r_1|=(d_2-r_1)+(d_1-r_2).\]
Consider the nested pairs $e_1'=(r_1,d_2)$ and $e_2'=(r_2,d_1)$. Note that $cost(\{e_1,e_2\})=cost(\{e'_1,e'_2\})$. Let $M'=(M\setminus\{e_1,e_2\})\cup \{e'_1,e'_2\}$. So, by swapping the pairs $e_1$ and $e_2$ with $e_1'$ and $e_2'$ the cost of matching does not change, $cost(M')=cost(M)$. 
Moreover,   $|S_{M'}|\leq |S_M|$, because any pair that is entwined with $e_1'$ or $e_2'$  is also entwined with either $e_1$ and $e_2$. Furthermore, $|I_{e_1'}|> |I_{e_1}|$. So, after swapping the pairs either $|S_{M'}|< |S_M|$ or  if $|S_{M'}|= |S_M|$ then the length of the largest interval in $S_{M'}$ is larger than $|I_{e'_1}|$. Each of these events can happen only for a finite number of times. Therefore,  the procedure described above terminates after some time. This implies that the set $S_{M'}$ becomes empty, and all the overlapping pairs will be nested.
\end{proof}

 \subsection{The walk on an (almost) balanced market}\label{sec: RW balanced}
  We continue to study the balanced market by translating matching cost into properties of the walk $W$, defined in \eqref{eq: walk}.
Let $M^*$ be a \textit{nested} optimal matching with penalties, which exists
by  \Cref{prop: nested}. 
Since the matching is nested, for any pair $(r,d)\in M^*$, there must be equal number of drivers and riders in $I_{(r,d)}$.  So, we have $W(r)=W(d)$.
 To analyze the distance between the points $r$ and $d$, we need the following definition. For $x\in[0,\ell]$, define the return time as the first time that the walk returns to $W(x)$ after at least one hop,
\begin{equation}\label{eq: defi tau}
\tau_x=\inf_{t>x}\{t:W(t)=W(x), \text{ and } H_t\neq H_x\},
\end{equation}
where $H_T$ was defined in \eqref{eq: definiiton H_T}.

Compare the definition of $\tau_0$ to the exit time  $\gamma_0$, as  defined in \eqref{eq: gamma} (see  \autoref{fig: tau vs gamma}). While $\gamma_0$ gave us an upper bound on the cost of matching, $\tau_0$ will give us a lower bound. 
In fact, for any  right oriented pair $(r,d)\in M^*$,  since $W(r)=W(d)$ we must have that $|d-r|\geq \tau_r$. 

In matching with penalties, we have the option to skip matching the rider $r$ by paying the cost $\nu$. So, the following result analyzes $\tau_0\wedge\nu$, where $x\wedge y=\min(x,y)$.
The idea is to first give a lower bound on the number of hops $H_{\tau_0}$, and then deduce a lower bound on  $\tau_0$.
\begin{lemma}\label{lm: tau lower bound}  Let $R=\{r_1,\ldots,r_n\}$ and $D=\{d_1,\ldots, d_m\}$ be two sets of random points drawn independently from a uniform distribution on $[0,\ell]$, where for some constant $c>0$, we have $0\leq m-n\leq c n^{1/2}$. Also, let $\tau_0$ be as in \eqref{eq: defi tau}, and $\nu\geq \frac{\ell}{n^{1/2-\delta}}$ for some $\delta\in(0,1/2]$.  
Then there exists some constant $C>0$ such that for large enough $n$
 \[\mathbb{E}[\tau_0\wedge \nu]\geq  C \ell n^{\frac{\delta}{2}-\frac{3}{4}}.\]
\end{lemma}
\begin{proof}
To prove the lemma, we start with bounding the number of hops in $[0,\tau_0]$. In fact, for any $1\leq k\leq\min(\frac{n}{2}, \frac{mn}{5(m-n)^2})$, we claim there exists a  constant $C>0$ 
such that,
\begin{equation}\label{eq: dist H_tau}
 \mathbb{P}[H_{\tau_0}=2k]\geq  \frac{C }{k\sqrt{k}}.
 \end{equation}

Consider the lattice path $P$, constructed from $W$ as in   \Cref{lm: dist hat H_i}. Assume that $H_{\tau_0}=2k$. Then the first part of $P$ is a path from $(0,0)$ to $(2k,0)$ that does not touch the $x$-axis except at the end-points. Such paths can either always stay above or below the $x$-axis, which implies the number of them is $\frac{2{2k\choose k}}{k+1}$. Then by adding any path from $(2k,0)$ to $(m+n-2k,m-n)$ we get a lattice path with $2k$ hops before the first return to zero.  Therefore,
\begin{align*}
    \mathbb{P}(H_{\tau_0}=2k)&=2\frac{{2k\choose k} {n+m-2k\choose n-k}}{(k+1){n+m\choose n}}\\
    & =2 \frac{{2k\choose k} }{(k+1)}\prod_{i=0}^{k-1}\frac{(n-i)(m-i)}{(n+m-2i)(n+m-2i-1)}\\
    &\geq \frac{2^{2k}}{(k+1)\sqrt {\pi k}}\prod_{i=0}^{k-1}\frac{1}{4} \big(1-\frac{m-n}{n-i}\frac{m-n}{m-i}\big)
\end{align*}
where the last  inequality is by Stirling's approximation and the following observation
\begin{align*}
\frac{(n-i)(m-i)}{(n+m-2i)(n+m-2i-1)}&\geq \frac{1}{(2+\frac{m-n}{n-i})(2-\frac{m-n}{m-i})}\\
&\geq\frac{1}{4}(1-\frac{m-n}{n-i}\frac{m-n}{m-i}).
\end{align*}
Using the inequality $\frac{m-n}{n-i}\frac{m-n}{m-i}\leq \frac{m-n}{n-k}\frac{m-n}{m-k}$,
we get,
\[    \mathbb{P}(H_{\tau_0}=2k)\geq\frac{2}{(k+1)\sqrt {\pi k}} \big(1-\frac{m-n}{n-k}\frac{m-n}{m-k}\big)^{k} .\]

Note that  if $k\leq \min(\frac{n}{2},\frac{1}{5}\frac{nm}{(m-n)^2})$ then
\[ \big(1-\frac{m-n}{n-k}\frac{m-n}{m-k}\big)^{k}\geq 1-k\frac{m-n}{n-k}\frac{m-n}{m-k}\geq 1-4k\frac{(m-n)^2}{mn}\geq \frac{1}{5}.\]
As a result, there exits a constant $C>0$ such that for any $k\geq 1$
\[ \mathbb{P}(H_{\tau_0}=2k)\geq \frac{C}{k\sqrt{k}},\]
which proves \eqref{eq: dist H_tau}.

Now, going back to the walk $W$, one can write
\begin{equation}\label{eq: equality hops and tau}
\mathbb{E}[\tau_0\wedge \nu]=\sum_{k=0}^\infty \mathbb{P}(H_{\tau_0}=2k)\mathbb{E}[\tau_0\wedge \nu|H_{\tau_0}=2k].
\end{equation}
Let $\ell_i$ be the position of the $i^{th}$ hop on $[0,\ell]$.  Then
\[\mathbb{E}[\tau_0\wedge \nu|H_{\tau_0}=2k]=\mathbb{E}[\ell_{2k}\wedge \nu|H_{\tau_0}=2k]=\mathbb{E}[\ell_{2k}\wedge \nu],\]
where the second equality holds because the length of a hop does not depend on whether it is a hop up or a hop down. 
On the other hand, since the $k^{th}$ order statistic of $n+m$ independent uniform samples  on $[0,1]$ is distributed as $Beta(k,n+m+1-k)$, using Chebyshev's inequality for $k\leq \frac{\nu(n+m)}{4\ell}$ and $\nu\geq \ell/n^{1/2-\delta}$ we get
\begin{align}\label{eq: bd ell 2k }
\begin{split}
    \mathbb{P}(\ell_{2k}\geq \nu)
    &\leq \frac{\frac{k\ell^2}{(n+m+1)^2}}{\frac{k\ell^2}{(n+m+1)^2}+(\nu-\frac{2k\ell}{n+m+1})^2}\\
    &\leq\frac{4k\ell^2}{\nu^2(n+m+1)^2} \leq \frac{4k}{n^{2\delta}(n+m+1)}.
    \end{split}
\end{align}
Let us denote $|x|_+=\max(0,x)$. Then using the above arguments for $k\leq \frac{n^{1/2+\delta}}{4}\leq \frac{\nu(n+m)}{4\ell}$, 
\begin{align*}
    \mathbb{E}[\ell_{2k}\wedge \nu]&= \mathbb{E}[\ell_{2k}]-\mathbb{E}|\ell_{2k}-\nu|_+\\
    &\geq \frac{2k\ell}{n+m+1}-|\ell-\nu|_+\mathbb{P}(\ell_{2k}\geq\nu)\\
    & \geq \frac{k\ell}{n+m+1}(2-\frac{4}{n^{2\delta}}),
\end{align*}
where the first inequality is by the fact that $|\ell_{2k}-\nu|\leq (\ell -\nu)$, and the second inequality is by \eqref{eq: bd ell 2k }. 
Since $\delta>0$, there exists a constant $\alpha>0$ such that $2-\frac{4}{n^{2\delta}}>\alpha$ for large enough $n$.
Combining the previous inequality with \eqref{eq: dist H_tau} and  \eqref{eq: equality hops and tau}  
\begin{align*}
    \mathbb{E}[\tau_0\wedge \nu]&\geq \alpha\sum_{k=0}^{\frac{n^{\delta+1/2}}{4}} \mathbb{P}(H_{\tau_0}=2k)\frac{k\ell}{n+m+1}\\
    &\geq\frac{\alpha\ell}{n+m+1}\sum_{k=1}^{\min(\frac{n}{2},\frac{nm}{5(n-m)^2},\frac{n^{\delta+1/2}}{4})}\frac{C}{\sqrt k}\\
   & \geq C'\ell\min(\frac{1}{\sqrt{n}},\frac{1}{m-n},n^{\frac{\delta}{2}-\frac{3}{4}}),
\end{align*}
where $C'$ does not depend on $n$, $\ell$. Now, using the condition that $m-n\leq cn^{1/2}$ we get the result.
\end{proof}

\subsection{Lower bound on the cost of matching}\label{sec: balanced proof}
This section is on the proof of  \Cref{thm: balanced}.
First, note that when $\nu\geq \ell$, we get a lower cost to match any riders with some rider than to pay the cost $\ell$. So,  Part \ref{part: balanced thm no cost} of  \Cref{thm: balanced} is equivalent to Part \ref{thm:balancedomn} of  \Cref{thm:balancedvsunbalanced}. So, we only give a proof for  \Cref{thm: balanced}.
Note also that, the lower bound in Part \ref{part: balanced thm no cost}  is a special case of  Part \ref{part: balanced thm with cost} when $\delta=\frac{1}{2}$. However, we give a separate proof by observing that in the balanced market the cost of the optimal matching (without any penalties) is equal to the area below the walk $W$, that has been studied in \citealt{RW-surface}.

\begin{proof}[Proof of  \Cref{thm: balanced}, Part \ref{part: balanced thm no cost}]
Given  the set of riders $R=\{r_1,\ldots,r_n\}$ and drivers $D=\{d_1,\ldots,d_n\}$ in the interval $[0,\ell]$,  let 
\[r_{(1)}\leq r_{(2)}\leq \cdots\leq r_{(n)},\quad\quad d_{(1)}\leq d_{(2)}\leq \cdots\leq d_{(n)},\]
be the order statistics of the sets $R$ and $D$.
We claim that pairing $r_{(i)}$  to $d_{(i)}$ is the minimum cost matching. Given this claim, by Proposition 2 in \citealt{RW-surface} and Stirling's approximation,
\[\mathbb{E}[cost( M^*)]=\mathbb{E}[\sum_{i=1}^n |r_{(i)}-d_{(i)}|]=\frac{n2^{2n-1}}{(2n+1){2n\choose n}}\ell \sim \frac{\ell\sqrt{\pi n}}{4},\]
which is the desired result.

 Let $M^*$ be the matching returned by omniscient.
Assume that $i$ is the first index which $r_{(i)}$ is not matched  to $d_{(i)}$ by omniscient. So, there exits some $j> i$ that $e_1=(r_{(i)},d_{(j)})\in M^*$. Since the market is balanced and $i$ is the first index that $(r_{(i)},d_{(i)})\not\in M^*$, there must exist some index $j'>i$  such that $e_2=(r_{(j')},d_{(i)})\in M^*$. Therefore, $e_1$ and $e_2$ are overlapping pairs, and we can swap them as in  \Cref{prop: nested} without increasing the cost of matching. As a result of the swap $r_{(i)}$ is matched  to $d_{(i)}$, and we can continue repeating this procedure.
\end{proof}

\begin{proof}[Proof of  \Cref{thm: balanced}, Part \ref{part: balanced thm with cost}]
Again,  let 
\[r_{(1)}\leq r_{(2)}\leq \cdots\leq r_{(n)},\quad\quad d_{(1)}\leq d_{(2)}\leq \cdots\leq d_{(n)},\]
be the order statistics of the sets $R$ and $D$. 
Next we give a lower bound on $\mathbb{E}[cost(M^*)]$, where $M^*$ is a \textit{nested} optimal matching with penalties, which we know it exists by  \Cref{prop: nested}. 
Let $e=(r,d)$ be a matched  pair in $M^*$. Since there is no unmatched rider or driver in $I_{e}$ and $M^*$ is nested, there must be equal number of riders and drivers in $I_e$.
So, if a rider at position $r$ is matched to driver on its right then $|I_{r,M^*(r)}|=\tau_r$,
where $\tau_r$  is defined \eqref{eq: defi tau}.
We claim that for $\frac{n}{3}\leq i\leq \frac{2n}{3}$, there is a positive probability that $r_{(i)}$ is either matched to a driver on its right or it is not matched at all. For that purpose we show that there exists some constant $\alpha>0$ such that  
$\mathbb{P}(-4\sqrt n\leq W_{r_{(i)}}\leq -3\sqrt{n})\geq \alpha$. 
To prove the claim, again we use lattice path presentation of the walk $W$, and we get
\begin{align}\label{eq: distance from expectation W}
\begin{split}
\mathbb{P}(-4\sqrt n\leq W_{r_{(i)}}\leq -3\sqrt{n})&=\sum_{k=3\sqrt{n}}^{4\sqrt n}\frac{{2i+k\choose i}{2n-2i-k\choose n-i}}{{2n\choose n}}\\
&\geq C\sum_{k=3\sqrt{n}}^{4\sqrt n}\frac{\sqrt n}{\sqrt{(2i+k)(2n-2i-k)}}e^{-\frac{k^2}{n}}\\
&\geq C\sum_{k=3\sqrt{n}}^{4\sqrt n}\frac{\sqrt n}{\sqrt{(\frac{4n}{3})^2-k^2}}e^{-16}
\geq C'\frac{3}{4},
\end{split}
\end{align}
where $C'$ is independent from $n$ and the second inequality is from equation (5.41) in \citealt{Asymptopia}. 
If $r_{(k)}$ is in a left oriented pair, then by $W_{r(k)}\leq -3\sqrt n$, there must exists at least $3\sqrt n$ unmatched riders on the left of $r_{(k)}$. The reason is non of them can match to a driver on the right of $r_{(k)}$ by  \Cref{prop: matching direction}. The cost of the unmatched riders in that case is at least $3\nu\sqrt n\geq n^{\delta}\ell$. So, we need to consider the events that $r_{(k)}$ is either unmatched or matched to driver on its right.  In this case, cost of matching $r_{(k)}$, denoted by $cost(r_{(k)})$, is at least $\min(\nu,\tau_{r_{(k)}})$.   Let the event $L_{r_{(k)}}$ indicate whether $r_{(k)}$ is in a left-oriented pair. Then
combining this observation with   \eqref{eq: distance from expectation W}, 
\begin{align*}
   \mathbb{E}[cost(M^*))]&\geq\sum_{i=\frac{n}{3}}^{\frac{2n}{3}}    \mathbb{E}[cost(r_{(k)})]\\
   &\geq C' \sum_{i=\frac{n}{3}}^{\frac{2n}{3}}    \mathbb{E}\big[cost(r_{(k)})\mid -4\sqrt n \leq W_{r_{(k)}}\leq -3\sqrt{n}\big] 
   \\
  &\geq C'\sum_{i=\frac{n}{3}}^{\frac{2n}{3}}\Big(\ell n^{\delta}\mathbb{P}(L_{r_{(k)}})+  \mathbb{E}\big[cost(r_{(k)})\mid -4\sqrt n \leq W_{r_{(k)}}\leq -3\sqrt{n}, L_{r_{(k)}}\big]\big(1-\mathbb{P}(L_{r_{(k)}})\big)\Big)
  \end{align*}
  Now, in the case $|W_{r(k)}|\leq 4\sqrt{n}$ we can apply  \Cref{lm: tau lower bound} to get
  \begin{align*}
  \mathbb{E}[cost(M^*))]&\geq C'\sum_{i=\frac{n}{3}}^{\frac{2n}{3}}\min\Big(\ell n^{\delta},  \mathbb{E}\big[\tau_{r_{(k)}}\wedge \nu\mid -4\sqrt n \leq W_{r_{(k)}}\leq -3\sqrt{n}\big]\Big)\\
   &\geq C''\min\Big(\ell n^{\delta}, \ell n^{\frac{1}{4}+\frac{\delta}{2}}\Big)\geq C''\ell n^{\delta},
\end{align*}
which is the desired result.
\end{proof}

\bibliographystyle{plainnat}
\bibliography{ref}




\end{document}